%% file: manuscript.tex
\documentclass[journal,doublecolumn]{IEEEtran}

\include{before_document}
\usepackage{lipsum}
\usepackage{makecell}
\usepackage{verbatim}

\makeatletter
\newenvironment{breakablealgorithm}
{
		\begin{center}
			\refstepcounter{algorithm}
			\hrule height.8pt depth0pt \kern2pt
			\renewcommand{\caption}[2][\relax]{
				{\raggedright\textbf{\ALG@name~\thealgorithm} ##2\par}%
				\ifx\relax##1\relax 
				\addcontentsline{loa}{algorithm}{\protect\numberline{\thealgorithm}##2}%
				\else 
				\addcontentsline{loa}{algorithm}{\protect\numberline{\thealgorithm}##1}%
				\fi
				\kern2pt\hrule\kern2pt
			}
		}{
		\kern2pt\hrule\relax
	\end{center}
}
\makeatother

\IEEEoverridecommandlockouts                              

\title{\LARGE \bf
	Resilient Multi-Dimensional Consensus and Distributed Optimization against Agent-Based and Denial-of-Service Attacks}
\author{Hongjian Chen, Changyun Wen, \textsl{Fellow, IEEE}, Xiaolei Li%
    \thanks{H. Chen and C. Wen are with the School of Electrical and Electronic Engineering, 
    Nanyang Technological University, Singapore 
    ({\tt\small HONGJIAN001@e.ntu.edu.sg, ecywen@ntu.edu.sg}).}%
    \thanks{X. Li is with the School of Electrical Engineering, 
    Yanshan University, Qinhuangdao, China ({\tt\small xiaolei@ysu.edu.cn}).}%
    \thanks{This manuscript is not published in any copyrighted journal, even though it was originally submitted to the \textit{IEEE Transactions on Systems, Man, and Cybernetics} on 05~Apr~2024, and with a decision of \textit{reject with major revision and encouragement to resubmit} received on 11~Sep~2024.}%
}

\begin{document}	
\maketitle
\thispagestyle{empty}
\pagestyle{plain}

\begin{abstract}
In this paper, we consider the resilient multi-dimensional consensus and distributed optimization problems of multi-agent systems (MASs) in the presence of both agent-based and denial-of-service (DoS) attacks. The considered agent-based attacks can cover malicious, Byzantine, and stubborn agents. The links between agents in the network can be blocked by DoS attacks, which may lead the digraph to time-varying and isolated. The objective is to ensure the remaining benign agents achieve consensus. To this end, an ``auxiliary point” based resilient control algorithm is proposed for MASs. Under the proposed algorithm, each healthy agent constructs a ``safe kernel" utilizing the states of its in-neighbors and updates its state toward a specific point within this kernel at each iteration. If an agent cannot receive its neighbors' states owing to DoS attacks, it will use the states received immediately before DoS attacks. Moreover, a resilient multi-dimensional distributed optimization (RMDO) algorithm is also proposed. Theoretical proofs and numerical examples are presented to demonstrate the effectiveness of the proposed algorithms.
\end{abstract}



\section{Introduction}

During recent decades, substantial research attention has been focused on the consensus and distributed optimization of MASs, the objective of which is to design distributed algorithms enabling agents to reach consensus \cite{ren2007information}, optimum \cite{ wang2023adaptive, hai2023novel} or Nash equilibrium seeking \cite{ye2023distributed} cooperatively. Many collaboration control schemes have been developed assuming that every agent is trustworthy and follows predefined algorithms. However, as scale and complexity of the network have increased, distributed control algorithms have become more prone to attacks. On the one hand, adversarial agents including malicious, Byzantine, or stubborn ones can drive the normal agents' states outside the desired region \cite{mo2011cyber}. On the other hand, attacks launched at the communication links, such as DoS attacks, can prevent information transmission over the network \cite{he2021secure}.

Considering the resilient consensus problems of MASs, some researchers have adopted the idea of detecting and isolating abnormal agents. Along this direction, observer-based techniques are used to develop a model-based detection system \cite{zhou2021resilient, silvestre2017stochastic}. Moreover, considering the computational cost, some resilient consensus algorithms within polynomial time are proposed in \cite{ramos2023discrete}. To enlarge the algebraic connectivity of the communication network without physically changing the topology, \cite{zhao2017resilient} and \cite{yuan2021secure} utilize two-hop broadcasting schemes to achieve the resilient scalar consensus, where each agent serves as a local monitor through two-hop information. \cite{luo2023secure} extends the above results to multi-dimensional systems with dynamics constraints and malicious attacks.

Another category of resilient control of MASs is designing distributed algorithms to ensure acceptable performance in the presence of false data. The earliest works propose distributed protocols by ignoring neighbors' most extreme values in a simple network (the complete network) \cite{dolev1986reaching}. Based on this idea, \cite{kieckhafer1994reaching} and \cite{vaidya2012iterative} develop a class of algorithms called the Mean Subsequence Reduced (MSR), where each benign agent discards the $F$ smallest and $F$ largest values and updates its state within a limited range with fixed lower and upper bounds. Instead of using the complete graph,  \cite{zhang2012robustness} and \cite{leblanc2013resilient} introduce a more general topology called the robustness network which provides redundancy information for each agent. With the help of this novel network, \cite{leblanc2013resilient} modifies MSR into Weighted Mean Subsequence Reduced (W-MSR). Unlike fixed bounds determined by MSR, W-MSR exclusively eliminates values that are strictly less than or greater than the healthy agents' states. This mechanism keeps more useful information than MSR and offers a graph-based analysis of the resiliency of MASs. \cite{wu2017secure} deploys the W-MSR under communication delays and \cite{dibaji2017resilient} extends the proposed algorithm into the asynchronous setting and second-order systems.

The class of algorithms based on MSR focuses on the resilient scalar consensus of MASs. Considering the resilient consensus in multi-dimensional space, Leblanc et al. \cite{leblanc2017resilient} implement the norm-based W-MSR to each coordinate of state vectors separately, driving the consensus of benign agents towards a multi-dimensional ``box”, the boundaries of which are confined by the minimum and maximum values of benign agents' states. Nevertheless, this ``box” typically encompasses a larger area than the desired region formed by the healthy agents' initial states. To tackle the problems mentioned above, \cite{xiang2016brief} and \cite{yan2020safe} develop a family of resilient multi-dimensional consensus algorithms via the Tverberg partitioning method, which partitions a sufficient number of points into subsets where the intersection of the convex hull of these subsets is non-empty. Take \cite{yan2020safe} as an example, each benign agent computes a ``safe kernel" by applying the Tverberg partitioning method and updates its state to any point in this kernel. The updating mechanism ensures that the final agreement will converge to a point in the convex hull formed by benign agents' initial states. The Tverberg partitioning-based algorithm is generalized in some resilient consensus-based problems such as rendezvous of multi-robot systems in \cite{park2017fault}, containment control in \cite{yan2020resilient} and synchronization of networked Lagrangian systems in \cite{chen2023resilient}. Considering a faster computation of Tverberg partitions, a quadratic programming (QP) method is proposed in \cite{wang2018resilient}. 

The resilient control scheme based on the ``auxiliary point" method \cite{yan2022resilient} is proposed to tackle consensus problems of MASs only under agent-based attacks. However, combining attacks launched on the communication links with agent-based attacks is not considered in the existing literature. Link-targeted attacks can break the digraph and prevent benign agents from receiving necessary information for computing the updating states as illustrated in \cite{9895458,9345784}. Consequently, the resilience and robustness of the ``auxiliary point" method can be degraded. Thus it is unclear whether the existing schemes including the algorithm in \cite{yan2022resilient} are still working in the presence of such attacks. In this paper, we will address such an issue with DoS attacks on communication links.

Furthermore, since the fundamental principle of consensus serves in many distributed coordination settings, the resilient consensus algorithm can be easily extended to tackle resilient distributed optimization problems in adversarial environments. Addressing the RMDO issues has the following difficulties: 1) Under a fully distributed setting, adversarial attacks can drive the global optimizer to any arbitrary point by simply tempering an adversary's local cost function according to Theorem 5 of \cite{pirani2023graph}. As a result, characterizing the region containing the true optimizer of the multi-dimensional global function is challenging without additional redundancy requirements, which can render the resilient distributed optimization problem unsolvable; 2) DoS attacks can break the robust network into many isolation parts and prevent benign agents from communicating with each other to achieve the global optimizer. Some recent attempts are now discussed. In \cite{gupta2021byzantine}, an exact convergence based on a norm filter in a complete network is achieved. Stochastic values of trust between agents are utilized in \cite{yemini2022resilience} to guarantee the exact fault tolerance. However, the requirements of the complete network and trust values are quite strong to some extent. Inspired by the above discussions, we consider how to relax such requirements, which is another problem to be solved. 

The main contributions together with the proposed approaches of this paper are summarized as follows:
\begin{itemize}
\item  By employing a policy that utilizes the immediate in-neighboring states before DoS attacks, the algorithm in \cite{yan2022resilient} is modified to mitigate both agent-based and edge-targeted DoS attacks. A rigorous analysis shows that, given a predefined upper bound on the number of adversarial agents and the duration of DoS attacks, our proposed algorithm ensures an agreement in MASs even without other knowledge of attacks.
\item We further extend our proposed resilient consensus algorithm to tackle the above challenges of RMDO under both agent-based and DoS attacks. By restricting the local cost functions of all agents and utilizing the subgradient method, the exact convergence to the global optimizer is achieved asymptotically. 
\end{itemize}

The subsequent sections of this paper are organized as follows: Section \uppercase\expandafter{\romannumeral2} presents preliminary results essential for understanding this study. Section \uppercase\expandafter{\romannumeral3} explains the research problems in detail. A resilient consensus algorithm and its performance analysis are provided in Section \uppercase\expandafter{\romannumeral4}. The RMDO algorithm and its performance analysis are provided in Section \uppercase\expandafter{\romannumeral5}. Section \uppercase\expandafter{\romannumeral6} presents numerical examples to demonstrate the effectiveness of proposed algorithms and the conclusion of this paper is given in section \uppercase\expandafter{\romannumeral7}.

\section{Preliminaries}
\subsection{Graph theory}
Consider an MAS consisting of $N$ agents cooperatively working on a directed graph $\mathcal{G}=\{\mathcal{V}, \mathcal{E}\}$, where $\mathcal{V}$ represents the set of agents, and $\mathcal{E} \subset \mathcal{V} \times \mathcal{V}$ denotes the set of edges. An edge $e_{i j} \in \mathcal{E}$ indicates that agents $i$ can receive agent $j$' state. Agent $i$'s in-neighbors and out-neighbors sets are $\mathcal{N}_i^{in}=\left\{j \in \mathcal{V} \mid e_{i j} \in \mathcal{E}\right\}$ and $\mathcal{N}_i^{out}=\left\{j \in \mathcal{V} \mid e_{j i} \in \mathcal{E}\right\}$.

To capture the ability of information to enter sets (via the edges) through individual vertices, a formal definition of robustness network is firstly introduced in \cite{leblanc2013resilient} as shown below:
\begin{definition} \label{r robust}
	 \cite{leblanc2013resilient} For any pair of subsets $\mathcal{V}_1, \mathcal{V}_2 \subsetneq \mathcal{V}$ satisfying $\mathcal{V}_1 \cap \mathcal{V}_2 = \emptyset$ and $\mathcal{V}_1, \mathcal{V}_2 \neq \emptyset$, the digraph $\mathcal{G}=\{\mathcal{V}, \mathcal{E}\}$ is classified as \textbf{r-robustness} if either of the following conditions holds:
	\begin{enumerate}[1)]
			\item There exists an agent $i \in \mathcal{V}_1$ has at least $r$ in-neighbors outside $\mathcal{V}_1$;
			\item There exists an agent $i \in \mathcal{V}_2$ has at least $r$ in-neighbors outside $\mathcal{V}_2$.
	\end{enumerate}
\end{definition}
\begin{definition} \label{rs robust}
	\cite{leblanc2013resilient} For any pair of subsets $\mathcal{V}_1, \mathcal{V}_2 \subsetneq \mathcal{V}$ satisfying $\mathcal{V}_1 \cap \mathcal{V}_2 = \emptyset$ and $\mathcal{V}_1, \mathcal{V}_2 \neq \emptyset$, the digraph $\mathcal{G}=\{\mathcal{V}, \mathcal{E}\}$ is classified as \textbf{(r,s)-robustness} if either or both of the following conditions hold:
	\begin{enumerate}[1)]
		\item Any agent $i \in \mathcal{V}_1$ has at least $r$ in-neighbors outside $\mathcal{V}_1$;
		\item Any agent $i \in \mathcal{V}_2$ has at least $r$ in-neighbors outside $\mathcal{V}_2$;
		\item There are no less than $s$ agents in $\mathcal{V}_1 \cup \mathcal{V}_2$, such that each of them has at least r in-neighbors outside $\mathcal{V}_1 \cup \mathcal{V}_2$.
	\end{enumerate}
\end{definition}

\subsection{Objective redundancy}
Even if the MAS is subject solely to agent-based attacks, the exact convergence to the global minimizer is impossible unless additional assumptions on the agents' local cost functions are made \cite{sundaram2018distributed}. Therefore, objective redundancy \cite{zhu2023resilient} is deployed to overcome this difficulty. Let $f_i(x): \mathbb{R}^d \to \mathbb{R}$ be the agent $i$'s multi-dimensional cost function. This paper adopts the following definitions and corollaries in \cite{zhu2023resilient}.

\begin{definition} \label{r-red-def}
    \cite{zhu2023resilient} A network consisting of $N$ agents is said to be $r$-redundant, $r\in \{0,1,...,N-1\}$, if for any subsets $\mathcal{V}_1, \mathcal{V}_2 \subsetneq \mathcal{V}$ with $|\mathcal{V}_1|=|\mathcal{V}_2|=N-r$, it holds
    \begin{align}
        \begin{gathered}
            \arg\min_{x} \sum_{i \in \mathcal{V}_1} f_i(x) = \arg\min_{x} \sum_{i \in \mathcal{V}_2} f_i(x).
        \end{gathered}
    \end{align}
\end{definition}

Let $\mathcal{X}^*$ denote the set of optimal solutions of the global cost function $f(x)=\frac{1}{N} \sum_{i \in \mathcal{V}}f_i(x)$. Based on Definition \ref{r-red-def}, \cite{zhu2023resilient} derived the following lemma.

\begin{lemma} \label{r-redundant}
    \cite{zhu2023resilient} For any subset $\mathcal{V}_1 \subsetneq \mathcal{V}$ with $|\mathcal{V}_1| \textcolor[RGB]{170.00,0.0,0.00}{\geq} N-r$, if an MAS consisting of $N$ agents is $r$-redundant, then it holds that
    \begin{align}
        \begin{gathered}
            \arg\min_{x} \sum_{i \in \mathcal{V}_1} f_i(x) = \mathcal{X}^*.
        \end{gathered}
    \end{align}
\end{lemma}

\subsection{Sarymsakov matrix}
Another important tool to handle the RMDO issues is the Sarymsakov matrix \cite{xia2014sarymsakov}. The row stochastic matrix $\mathcal{A}=\{a_{ij}\}$ is associated with the digraph $\mathcal{G}$, where $a_{ij}>0$ if and only if $e_{ij} \in \mathcal{E}$ and $\sum_{j=1}^{N} a_{ij} = 1$. Before giving the concept of the Sarymsakov matrix, \textit{one-stage consequent indices}  \cite{seneta1979coefficients} is introduced:
\begin{align*}
        \mathcal{I}_\mathcal{A}(\mathcal{V}_1) = \{ j : a_{ij} > 0 \text{ for some } i \in \mathcal{V}_1 \}.
\end{align*}
$\mathcal{I}_\mathcal{A}(\mathcal{V}_1)$ denotes the set of agents that directly influence those within $\mathcal{V}_1 \subsetneq \mathcal{V}$. Based on one-stage consequent indices, the Sarymsakov matrix is defined as follows:
\begin{definition} \label{Sary}
   For any disjoint nonempty sets $\mathcal{V}_1, \mathcal{V}_2 \subsetneq \mathcal{V}$, if one of the following statements hold:
    \begin{enumerate} [1)]
        \item $\mathcal{I}_\mathcal{A}(\mathcal{V}_1) \bigcap \mathcal{I}_\mathcal{A}(\mathcal{V}_2) \neq \emptyset$;
        \item $\mathcal{I}_\mathcal{A}(\mathcal{V}_1) \bigcap \mathcal{I}_\mathcal{A}(\mathcal{V}_2) = \emptyset$ and \\ $|\mathcal{I}_\mathcal{A}(\mathcal{V}_1) \bigcup \mathcal{I}_\mathcal{A}(\mathcal{V}_2)|>|\mathcal{V}_1 \bigcup \mathcal{V}_2|$,
    \end{enumerate}
    then, a row stochastic matrix $\mathcal{A}$ is classified as Sarymsakov.
\end{definition}
\begin{remark}
    Definition \ref{Sary} indicates that sets $\mathcal{V}_1$ and $\mathcal{V}_2$ either have some influencing nodes in common, or in the absence of common influencing nodes, the number of influencing agents surpass that of the agents being influenced. The definition of Sarymsakov can be easily satisfied in a robustness work according to Definition \ref{r robust} and \ref{rs robust}. 
\end{remark}

\section{Problem formulation}

Suppose there are $N$ agents that cooperate over a digraph $\mathcal{G}=\{\mathcal{V}, \mathcal{E}\}$. Let $x_i (t_k)\in \mathbb{R}^d$ be the state of agent $i$ at instant $t_k$, where $t_{k+1}-t_k=T$ with an arbitrary positive constant $T$ being the sampling period. The system model of benign agents is as follows:
\begin{align}
    \begin{gathered}\label{update}
        x_i(t_{k+1}) = x_i(t_k) + u_i(t_k) + \varepsilon_i(t_k),
    \end{gathered}
\end{align}
where $u_i(t_k)$ is the control input to be designed in the following section and the residual $\varepsilon_i(t_k)$ can be regarded as a disturbance satisfying the following assumption. 

\begin{assumption} \label{residual assumption}
    The residual term in the updating protocol satisfies:
    \begin{align}\label{residual}
                \sum_{t_0}^{\infty}\left\| \varepsilon_i(t_k)\right\|<\infty.
        \end{align}
\end{assumption}

\subsection{Agent-based attacks}
 We characterize the scope of agent-based attacks in this subsection. Let $\mathcal{F}$ and $\mathcal{B}$ be the set of adversarial and benign agents, which satisfying $\mathcal{F} \cap \mathcal{B}=\emptyset  $ and $ \mathcal{F} \cup \mathcal{B}=\mathcal{V}$. Denoting $F$ as the maximum fault tolerance of the digraph, two models of agent-based attacks are outlined as follows:
\begin{enumerate}
	\item \textsl{F}-total attack: The maximum of adversarial agents is $F$, i.e. $|\mathcal{F}| \le F$.
	\item \textsl{F}-local attack: The number of adversarial agents in the in-neighborhood of any benign agent $i$ is no more than $F$, i.e. $|\mathcal{F} \cap \mathcal{N}_i^{in}| \leq F$.
\end{enumerate}

Agent-based attacks can include stubborn, malicious, and Byzantine agents and there's no restriction on information transmitted by adversarial agents. With a maximum tolerance of adversarial agents, an assumption associated with the robustness network is adopted as shown below:
\begin{assumption} \label{No of in neighbors}
	The number of in-neighbors of any agent satisfies $\left|\mathcal{N}_i^{in}\right| \geq(d+1) F+1$, where $d$ denotes the state's dimension.
\end{assumption}
\begin{remark}
    Assumption 2 explicitly describes redundancy information needed for benign agents. In practical scenarios, $|\mathcal{N}_i^{in}|$ can be easily determined by the dimension $d$ and the network's maximum faulty tolerance of $F$.
\end{remark}

\subsection{DoS attacks}
The edge-targeted attacks considered are DoS attacks. Specifically, DoS attack strategies are launched at the communication links independently and the multiple strategies have a joint effect on the entire communication network. Here, we denote the $m$-th DoS attack at the edge $e_{i j}$ as $\mathcal{I}_{m}^{ij} = [ h_{m}^{ij}, h_{m}^{ij} + \tau_{m}^{ij} ]$, where $h_{m}^{ij}$ is the start instant and $\tau_{m}^{ij}$ is the duration of $m$-th attack. For any $t'\in \mathcal{I}_{m}^{ij}$, the information transmission from agent $i$ to agent $j$ will be interrupted. Let $\Pi_{D}^{ij} (t_{1}, t_{2}) \triangleq (t_{1}, t_{2}) \cap \bigcup_{m=1}^{\infty} \mathcal{I}_{m}^{ij}$ be the set where $e_{i j}$ is blocked during $(t_1,t_2)$. The entire DoS strategy is defined as

\begin{align}\label{DoS}
	\begin{gathered}
		\Pi_{D} (t_{1}, t_{2}) \triangleq (t_{1}, t_{2}) \bigcap \bigcup_{e_{ij} \in \mathcal{E}} \Pi_{D}^{ij} (t_{1}, t_{2}),
	\end{gathered}
\end{align}
where $\Pi_{D} (t_{1}, t_{2})$ is the set in which at least one edge is blocked. The commonly accepted assumption for DoS attacks' duration \cite{xu2019distributed,deng2020mas} is listed as follows:
\begin{assumption}\label{DoS-duration}
	(DoS duration) 
	\begin{align}
		\begin{gathered}
			|\Pi_{D} (t_{1}, t_{2})| \leq \mu_d + \frac{t_2 - t_1}{T_{d}},  \forall t_2 > t_1 \geq 0,
		\end{gathered}
	\end{align}
where constants $T_{d}>1$ and $\mu_d>0$.
\end{assumption}

\begin{remark}
	Assumption \ref{DoS-duration} is reasonable because the energy available for multiple adversaries is finite.
\end{remark}

\subsection{Resilient consensus problem}
Considering the influence of agent-based and DoS attacks, an algorithm is to be designed such that the resilient multi-dimensional consensus is achieved in the sense that the following conditions are satisfied:
\begin{enumerate}

	\item \textsl{Consensus}: For any agents $i, j \in \mathcal{B}$, it holds that
	\begin{align}\label{Consensus}
		\begin{gathered}
			\lim _{k \rightarrow \infty}[x_i\left(t_k\right)-x_j\left(t_k\right)]=0.
		\end{gathered}
	\end{align}
    \item \textsl{$\delta$-Validity}: The state of any benign agent remains in the set $\Xi(t_0)+\mathbb{B}_{\delta}$ at each iteration, where $\Xi(t_0)$ denotes the convex hull formed by all benign agents' initial states and $\mathbb{B}_\delta$ is $d$-dimensional ball with radius $\delta$.
\end{enumerate}

\subsection{RMDO problem}
In this paper, the convex but not necessarily differentiable local cost function $f_i(x)$ is only known by agent $i$. With only local objective information, multiple agents cooperatively solve the following RMDO problem:
\begin{enumerate}
    \item \textsl{Consensus}: It is the same as illustrated in (\ref{Consensus}).
    \item \textsl{Optimization}: 
    \begin{align} \label{optimization target}
        \begin{gathered}
            \lim_{k \rightarrow \infty}x_i(t_k) \in \mathcal{X}^*, 
        \end{gathered}
    \end{align}
    for any $i\in \mathcal{B}$. $\mathcal{X}^*$ denotes the set of optimal states satisfying $\mathcal{X}^* = \underset{x}{\arg \min } \sum_{i \in \mathcal{V}} f_i(x).$
\end{enumerate}

\section{A resilient consensus algorithm}

\subsection{Algorithm design}
The resilient consensus algorithm against both agent-based and DoS attacks proposed in this paper is based on the approach ``auxiliary point" in \cite{yan2022resilient}, mainly derived from ``safe kernel" introduced in \cite{yan2020safe}. The calculation process of the ``safe kernel" is discussed in the following definitions.

\begin{definition}\label{subset convex}
	$\mathcal{X}_i(t_k) \subset \mathbb{R}^d$ with cardinality $|\mathcal{N}_i^{in} |$ aggregates all in-neighboring states of agent $i\in \mathcal{B}$ at instant $t_k$. Given some positive integers $F<\left|\mathcal{N}_i^{in}\right|$, $\mathcal{S}(\mathcal{X}_i (t),F)$ denotes the set of all its subsets with cardinality $|\mathcal{N}_i^{in}|-F$.
\end{definition}
The set $\mathcal{S}(\mathcal{X}_i(t_k), F)$ encompasses $\left(\begin{array}{c}\left|\mathcal{N}_i^{in}\right| \\ F\end{array}\right)$ elements, each corresponding to a distinct convex hull. The ``safe kernel" is defined by the intersection of these convex hulls as follows:
\begin{definition}\label{safe kernel}
	\cite{yan2020safe} Given a set $\mathcal{X}_i(t_k) \subset \mathbb{R}^d$ with cardinality $\left|\mathcal{N}_i^{in}\right|$, for some constants $F \in \mathbb{Z}_{\geq 0}$ satisfying  $F<\left|\mathcal{N}_i^{in}\right|$, we define the intersection as the set
	\begin{align}
		\Psi\left(\mathcal{X}_i(t_k), F\right) \triangleq \bigcap_{S \in \mathcal{S}\left(X_i(t), F\right)} \operatorname{Conv}\left(S\right),
	\end{align}
 where $\operatorname{Conv}(S)$ is the convex hull formed by points in each subset $S \subset \mathbb{R}^d$.
\end{definition}


Then, the ``auxiliary point" is selected within the ``safe kernel". 

Now the resilient consensus algorithm is presented as Algorithm 1, where $\mathcal{H}_{\mathcal{X}}$ denotes a multi-dimensional convex hull and $\text{Cen}(\mathcal{X})$ is the centroid of $\mathcal{H}_{\mathcal{X}}$.

\begin{breakablealgorithm}
	\caption{Resilient consensus algorithm}
	\label{al1}
	\begin{algorithmic}[1]
		\STATE Collect all in-neighboring states $x_j(t_k), j \in \mathcal{N}_i^{in}$ in set $\mathcal{X}_i(t_k)$. \\
        \textbf{if} $t_k \in \mathcal{I}_{m}^{ij}$ \textbf{then} \\
        \hspace{\algorithmicindent} $x_j(t_k)=x_j(t_k^{-})$, where $t_k^{-} \leq h_m^{ij}$ is the last updating \\ \hspace{\algorithmicindent} instant before $m$-th DoS attack $\mathcal{I}_{m}^{ij}$. \\
        \textbf{end if}
		\STATE \textbf{for} $p\in \{1, 2, ..., d\}$ \textbf{do} \\
        \hspace{\algorithmicindent} \textbf{\romannumeral 1:} Sort the points of $\mathcal{X}_i(t_k)$ in ascending order based \\
        \hspace{\algorithmicindent} on their $p$-th entries.\\
        \hspace{\algorithmicindent} \textbf{\romannumeral 2:} Collect the first 
         and last $(d+1)F+1$ sorted points \\
        \hspace{\algorithmicindent} in sets $\mathcal{Y}_i(p,t_k)$ and  $\mathcal{Z}_i(p,t_k)$ respectively. Calculate \\ 
        \hspace{\algorithmicindent} any points $y_i(p,k) \in \Psi\left(\mathcal{Y}_i(p,t_k), F\right)$ and  $z_i(p,k) \in$ \\ \hspace{\algorithmicindent} $\Psi\left(\mathcal{Z}_i(p,t_k), F\right)$.\\
        \textbf{end for}
		\STATE Collect $y_i(p,k), 1\leq p \leq d$ and $z_i(p,k), 1\leq p \leq d$ in a set $\Lambda_i(t_k)$. The ``auxiliary point" of agent $i$ is chosen as the center of the convex hull $\mathcal{H}_{\Lambda_i(t_k)}$ formed by the set $\Lambda_i(t_k)$
        \begin{align}
                \tilde{x}_i(t_k)=\text{Cen}(\Lambda_i(t_k)).
        \end{align}
	\STATE The control input $u_i(t_k)$ is designed as follows
            \begin{align} \label{control input}
                 u_i(t_k) = (1-\alpha_i)(\tilde{x}_i(t_k)-x_i(t_k)),
            \end{align}
  where the weight is chosen as $\alpha_i < c$, $1-\alpha_i < c$, and constant $0.5 < c <1$.
		\STATE Transmit the updated state $x_i(t_{k+1})$ to out-neighbors $j\in \mathcal{N}_i^{out}$.
	\end{algorithmic}
\end{breakablealgorithm}

\begin{remark}
    Normally, the in-neighboring states of benign agents are treated as $0 \in \mathbb{R}^d$ if DoS attacks block edges such as in \cite{an2018decentralized}, as benign agents receive no information in such cases. This may result in the maximum fault tolerance $F$ of the ``auxiliary point" method being exceeded. Now this is not the case in our proposed algorithm, as seen from the policy in Step 1 which is inspired by the idea behind switching control. This means that if the edge $e_{ij}$ is blocked by $m$-th DoS attack $\mathcal{I}_{m}^{ij}$, agent $i$ will use the previous information of in-neighbor $j$. This policy allows the ``auxiliary point" approach to eliminate the influence of DoS attacks, which is proved rigorously in the next subsection. 
    By applying (\ref{control input}) into the updating protocol (\ref{update}), we have the updating protocol:
    \begin{align}\label{updating-complete}
        x_i(t_{k+1}) = \alpha_i x_i(t_k) + (1-\alpha_i)\tilde{x}_i(t_k) + \varepsilon_i(t_k).
    \end{align}
    Moreover, The designed weight $\alpha_i$ is used to ensure the sufficient effect of ``auxiliary point" $\tilde{x}_i(t_k)$ according to \cite{yan2022resilient} and a QP-based method in \cite{wang2018resilient} can be employed to reduce the computation cost for calculating $y_i(p,k)$ and $z_i(p,k)$.
\end{remark}
\begin{remark}
The weight $\alpha_i$ can be designed as either time-invariant or time-varying, based on practical considerations. 
As established in Lemma 3, the proposed algorithm ensures asymptotic consensus under both settings, provided that $\alpha_i$ satisfies the design guideline. 
It is worth noting that the convergence rate is influenced by the constant $c$ as discussed in Remark 8.
\end{remark}

\subsection{Performance analysis}
The reliability of Algorithm 1 is guaranteed by considering the existence of $y_i\left(p,t_k\right)$ and $z_i\left(p,t_k\right)$ for any dimension $p\in \{1,\ 2, \ ...,\ d\}$. 

Consider the number of $\mathcal{Y}_i(p,t_k)$'s elements, with the application of Corollary 1 in \cite{yan2022resilient} induced from \textbf{Helly's Theorem}, we can derive that the set $\mathcal{Y}_i(p,t_k)$ is nonempty. Thus, the point $y_i(p,t_k)$ must exist. Similarly, we can prove the existence of $z_i(p,t_k)$.

The following Lemma considers the $\delta$-validity condition by applying Algorithm 1 under both agent-based and DoS attacks. 
\begin{lemma}\label{vanish}
	Consider an MAS cooperating over the \textbf{(d + 1)F + 1}-robust [\textbf{(dF + 1, F + 1)}-robust] digraph subject to DoS attacks. If the agent-based attacks conform to the F-local [F-total \textbf{resp.}] attack model and Assumptions 1-2 hold, Algorithm 1 ensures the fulfillment of the $\delta$-validity condition.
\end{lemma}

\begin{proof}
    For simplicity, we denote the ``safe kernel" $\Psi\left(\mathcal{X}_i(t_k), F\right)$ as $\mathcal{S}_i(t_k)$. Also, we define $\delta(t_k)$ as
\begin{align}\label{delta}
    \begin{gathered}
        \delta(t_k) \triangleq \max_{i \in \mathcal{B}}\|\varepsilon_i(t_k)\|.
    \end{gathered}
\end{align}
If the MAS is only under $F$-total or $F$-local attacks, for any agent $i\in \mathcal{B}$, it has at least $|\mathcal{X}_i(t_k)|-F$ benign in-neighbors. From Definitions \ref{subset convex} and \ref{safe kernel} and Assumption \ref{No of in neighbors}, we can easily find that $\mathcal{S}_i(t_k)$ lies in the convex hull formed by any $|\mathcal{X}_i(t_k)|-F$ in-neighboring points. As a result, we have $\mathcal{S}_i(t_k) \in \Xi(t_k)$. Given that $\mathcal{Y}_i(p,t_k)$ is a subset of the $\mathcal{X}_i(t_k)$, it directly follows from Lemma 9 in \cite{yan2022resilient} that $\Psi(\mathcal{Y}_i(p,t_k)) \subseteq \mathcal{S}_i(t_k)$. Consequently, this implies $y_i(p, t_k) \in \mathcal{S}_i(t_k)$, and similarly, $z_i(p,t_k) \in \mathcal{S}_i(t_k)$. This gives that 
\begin{align} \label{safe kernel in Xi(t_k)}
    \Lambda_i(t_k) \subset \Xi(t_k),
\end{align}
where $\Xi(t_k)$ represents the convex hull formed by states of all benign agents at $t_k$.
 Since the \textcolor[RGB]{170.00,0.0,0.00}{``}auxiliary point\textcolor[RGB]{170.00,0.0,0.00}{"} $\tilde{x}_i(t_k)$ is the center of $\Lambda_i(t_k)$, we have $\tilde{x}_i(t_k)\in \Xi(t_k)$. Combined with $x_i(t_k) \in \Xi(t_k)$ and the updating protocol (\ref{updating-complete}), it follows that $x_i(t_{k+1})$ also belongs to $\Xi(t_k)$ in the absence of disturbances and DoS attacks. Furthermore, it is established that $x_i(t_{k+1})$ constitutes a vertex of the set $\Xi(t_{k+1})$. Consequently, this implies that $\Xi(t_{k+1})\subset \Xi(t_k)$, indicating the vanishing property of $\Xi(t_k)$ over time.

    Then, we consider the effect of the disturbance, namely, residual terms. For any benign agent $i$, we have
    \begin{align}
        \begin{gathered}
            x_i(t_{k+1}) \in \{x: \|x-\chi\| \leq \delta(t_k), \exists \chi \in \Xi(t_k)\},
        \end{gathered}
    \end{align}
    where $\delta(t_k)$ is defined as (\ref{delta}). So far, we have that if only agent-based attacks exist, the following relation hold
    \begin{align} \label{Xi subset}
        \begin{gathered}
            \Xi\left(t_{k+1}\right) \subset \Xi\left(t_k\right)+\mathbb{B}_{\delta(t_k)}.
        \end{gathered}
    \end{align}
    
    If the network is under DoS attacks at instant $t_k$, we consider a benign agent under the worst case, that is all $e_{ij}, i\in \mathcal{B}$ and $j\in \mathcal{N}_i^{in}$ are blocked by $m$-th DoS attack. The following relation of the in-neighboring set will hold
    \begin{align} \label{DoS_in_neighbor}
        \begin{gathered}
            \mathcal{X}_i(t_k) = \mathcal{X}_i(t_k^-) = \mathcal{X}_i(t_k+\tau_m),
        \end{gathered}
    \end{align}
    where $t_k^-$ denotes the last updating instant before $m$-th DoS attack $\mathcal{I}_{m}^{ij}$ and we define $\tau_m \triangleq \tau_m^{ij}.$ According to (\ref{DoS_in_neighbor}), the benign agent $i$ under the worst case does not update its state during $m$-th DoS attack. Further, considering the worst case for the whole network, all benign agents are under DoS attacks, with the policy given in Step 1 we have
    \begin{align}
        \begin{gathered}
            \Xi\left(t_{k+1}\right) = \Xi\left(t_k\right)+\mathbb{B}_{\delta(t_k)}.
        \end{gathered}
    \end{align}
    Combined with (\ref{Xi subset}), we have the vanishing property of $\Xi(t_k)$, namely
    \begin{align}\label{Xi_DoS-Ad}
		\Xi\left(t_{k+1}\right) \subseteq \Xi\left(t_k\right)+\mathbb{B}_{\delta(t_k)}.
	\end{align}
    From the definition of $\delta(t_k)$ and Assumption \ref{residual assumption}, the following relation holds
    \begin{align} \label{residual max}
        \begin{gathered}
            \sum_{t_0}^{\infty} \delta(t_k)<\sum_{i \in \mathcal{B}}\left(\sum_{t_0}^{\infty}\left\|\varepsilon_i(t_k)\right\|\right)<\infty.
        \end{gathered}
    \end{align}
    Combining (\ref{Xi_DoS-Ad}) and (\ref{residual max}), we conclude that $x_i(t_k)\in \Xi(t_0)+\mathbb{B}_{\delta}$ holds for any $t_k$.
\end{proof}

The subsequent lemma discusses the agreement condition (\ref{Consensus}) with the application of Algorithm 1.

\begin{lemma} \label{lemma agreement}
	Consider an MAS cooperating over the \textbf{(d + 1)F + 1}-robust [\textbf{(dF + 1, F + 1)}-robust] digraph subject to DoS attacks and suppose the agent-based attacks conform to the F-local [F-total \textbf{resp.}] attack model. With Assumptions 1-3, Algorithm 1 ensures that the consensus condition (\ref{Consensus}) is achieved asymptotically.
\end{lemma}

\begin{proof}
The minimum and maximum values along the dimension $p\in \{1,\ 2,\ ...,\ d \}$ are respectively denoted as follows
\begin{align*}
\begin{gathered}
    	m_{\mathcal{B}}^p\left(t_k\right)  \triangleq \min _{i \in \mathcal{B}} x_i^p\left(t_k\right),
	 \\ M_{\mathcal{B}}^p\left(t_k\right)  \triangleq \max _{i \in \mathcal{B}} x_i^p\left(t_k\right).
\end{gathered}
\end{align*}
The difference between them is defined as
\begin{align*}
	\Delta^p(t_k) \triangleq  M_{\mathcal{B}}^p\left(t_k\right)-m_{\mathcal{B}}^p\left(t_k\right).
\end{align*}
The relation (\ref{Xi_DoS-Ad}) implies the following inequalities:
\begin{align*}
   \begin{gathered}
       M_{\mathcal{B}}^p\left(t_{k+1}\right)  \leq  M_{\mathcal{B}}^p\left(t_k\right) + \delta(t_k),\\
        m_{\mathcal{B}}^p\left(t_{k+1}\right)  \geq   m_{\mathcal{B}}^p\left(t_k\right) - \delta(t_k).
   \end{gathered}
\end{align*}
Construct two sets as follows
\begin{align*}
    \begin{gathered}
        \mathcal{B}^p_M(t_k, \epsilon) \triangleq \{i \in \mathcal{B}: \ x_i^p(t_k)>\epsilon\},\\
        \mathcal{B}^p_m(t_k, \epsilon) \triangleq \{i \in \mathcal{B}: \ x_i^p(t_k)<\epsilon\}.
    \end{gathered}
\end{align*}

Given the symmetry across all dimensions, the subsequent analysis of the agreement condition will be simplified by examining only the first component. Suppose at some initial instants, it holds that $M_{\mathcal{B}}^p\left(t_k\right) \neq m_{\mathcal{B}}^p\left(t_k\right)$. Define $\epsilon_0 = \Delta^1(t_k)/2$, we will have two disjoint and nonempty sets $\mathcal{B}^p_M(t_k, M_{\mathcal{B}}^p\left(t_k\right)-\epsilon_0)$ and $\mathcal{B}^p_m(t_k, m_{\mathcal{B}}^p\left(t_k\right)+\epsilon_0)$. And the above two sets satisfy (\textit{4a}) and (\textit{4b}) topology conditions of Theorem 1 in \cite{yan2022resilient} with Assumption \ref{No of in neighbors}. For agent $j \in \mathcal{B}_M^1(t_k, M_{\mathcal{B}}^1\left(t_k\right)-\epsilon_0)$ we can conclude that the upper bound of the first entry of $\Lambda_j(t_k)$ is $M_{\mathcal{B}}^1\left(t_k\right)-\epsilon_0$. For simplicity, we rewrite $M_{\mathcal{B}}^1\left(t_k\right)$ and $m_{\mathcal{B}}^1\left(t_k\right)$ as $M_{\mathcal{B}}^1$ and $m_{\mathcal{B}}^1$ respectively. Then we have the first entry of agent $j$'s ``auxiliary point" satisfies
\begin{equation*}
    \begin{split}
        \tilde{x}^1_j(t_k)  &=\text{Cen}(\Lambda_j(t_k))\\
                          &=0.5 m_1\left(\Lambda_j(t_k)\right)+0.5 M_1\left(\Lambda_j(t_k)\right) \\  
                          &\leqslant 0.5\left(M^1_B-\epsilon_0\right)+0.5 M^1_B \\  
                          &=M^1_B-0.5 \epsilon_0.
    \end{split}
\end{equation*}

Then, considering the worst effect of DoS attacks on agent $j$, the first entry of the updated state fulfills the following condition
\begin{equation} \label{upperbound_j_k+1}
    \begin{split}
         x^1_j(t_{k+1+\tau_m}) & =\alpha_j x^1_j(t_k)+\left(1-\alpha_j\right) \tilde{x}^1_j(t_k)+\varepsilon^1_j(t_k) \\ 
                        & \leqslant \alpha_j M^1_\mathcal{B}+\left(1-\alpha_j\right)\left(M^1_\mathcal{B}-0.5 \epsilon_0\right) +\delta(t_k) \\ 
                        & =M^1_\mathcal{B}-0.5\left(1-\alpha_j\right) \epsilon_0+\delta(t_k) \\ 
                        & \leqslant M^1_\mathcal{B}-0.5 c \epsilon_0+\delta(t_k).
    \end{split}
\end{equation}

Consider a benign agent $i$ not satisfying $x_i^1(t_k) < M^1_\mathcal{B}(t_k) - \epsilon_0$, that is, $i\in \mathcal{B} \setminus \mathcal{B}_M^1(t_k, M_{\mathcal{B}}^1\left(t_k\right)-\epsilon_0)$. We also assume the worst case of DoS attacks where all in-neighboring edges are blocked by DoS attacks. The following relation holds
\begin{equation} \label{upperbound_i_k+1}
    \begin{split}
        x_i^1(t_{k+1+\tau_m}) & =\alpha_i x_i^1(t_k) + (1-\alpha_i) \tilde{x}^1_i(t_k)+\varepsilon^1_i(t_k) \\
                        & \leqslant \alpha_i (M^1_\mathcal{B}-\epsilon_0)+\left(1-\alpha_j\right)M^1_\mathcal{B} + \delta(t_k) \\ 
                        & \leqslant M^1_\mathcal{B}-c \epsilon_0+\delta(t_k).
    \end{split}
\end{equation}

Obviously, the upper bound (\ref{upperbound_j_k+1}) can be also applied to (\ref{upperbound_i_k+1}). Similarly, for $j \in \mathcal{B}_m^1(t_k, m_{\mathcal{B}}^1-\epsilon_0)$, we can induce that $ x^1_j(t_{k+1}++\tau_m)\geqslant m^1_\mathcal{B}+0.5 c \epsilon_0-\delta(t_k),$ and this lower bound can be applied to $i\in \mathcal{B} \setminus \mathcal{B}_m^1(t_k, m_{\mathcal{B}}^1+\epsilon_0)$. Define $\epsilon_1 \triangleq 0.5c\epsilon_0 - \delta(t_k)$, where $\epsilon_1<\epsilon_0$ holds. The previous analysis on the worst case of DoS attacks indicates the following situations:
\begin{itemize}
    \item At least one agent in $\mathcal{B}_M^1(t_k, M_{\mathcal{B}}^1-\epsilon_0)$ has its first entry's upper bound decreasing from $M_{\mathcal{B}}^1$ to below or equal to $M_{\mathcal{B}}^1-\epsilon_0$, which leads to $\mathcal{B}_M^1(t_{k+1}+\tau_m, M_{\mathcal{B}}^1-\epsilon_0) \subseteq \mathcal{B}_M^1(t_k, M_{\mathcal{B}}^1-\epsilon_0)$.
    \item At least one agent in $\mathcal{B}_m^1(t_k, m_{\mathcal{B}}^1+\epsilon_0)$ has its first entry's lower bound increasing from $m_{\mathcal{B}}^1$ to above or equal to $m_{\mathcal{B}}^1-\epsilon_0$, which leads to $\mathcal{B}_m^1(t_{k+1}+\tau_m, m_{\mathcal{B}}^1+\epsilon_0) \subseteq \mathcal{B}_m^1(t_k, m_{\mathcal{B}}^1+\epsilon_0)$.
\end{itemize}
For simplicity, we denote $M_{\mathcal{B}}^1-\epsilon_0$ as $\sigma_M$ and $m_{\mathcal{B}}^1+\epsilon_0$ as $\sigma_m$. From the former discussion, it holds that
\begin{equation} \label{no of B vanish}
    \begin{split}
        & \left|\mathcal{B}_M^1\left(t_{k+1}+\tau_m, \sigma_M\right)\right|+\left|\mathcal{B}_m^1\left(t_{k+1}+\tau_m, \sigma_m\right)\right| \\ & \leq \left|\mathcal{B}_M^1\left(t_k, \sigma_M\right)\right|+\left|\mathcal{B}_m^1\left(t_k, \sigma_m\right)\right|,
    \end{split}
\end{equation}
which is also illustrated in Fig. \ref{explanation of the former discussion}.
\begin{figure}[h]
	\centering
	\includegraphics[scale=0.3]{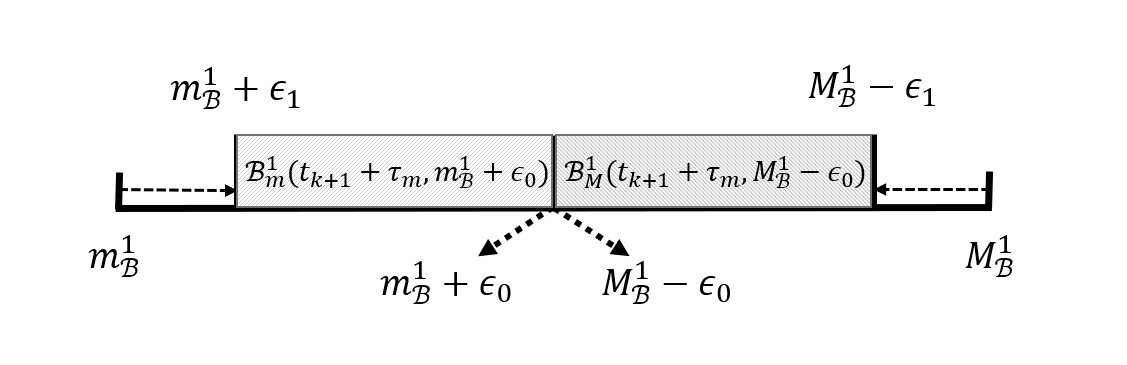}
	\caption{An explanation of the former discussion and  (\ref{no of B vanish}).}
	\label{explanation of the former discussion}
\end{figure}

The sets $\mathcal{B}_M^1\left(t_{k+1}+\tau_m, \sigma_M\right)$ and $\mathcal{B}_m^1\left(t_{k+1}+\tau_m, \sigma_m\right)$ are disjoint and it is assumed that both sets are nonempty. Based on it, we further consider the next updating step after $m$-th DoS attacks. Similarly, we can find a benign agent $j \in \mathcal{B}_M^1\left(t_{k+1}+\tau_m, \sigma_M\right)$ under the worst case of $(m+1)$-th DoS attacks  with the following relation
\begin{equation*}
    \begin{split}
        \tilde{x}^1_j(t_{k+2}+\tau_m+\tau_{m+1}) &=0.5 m_1\left(\Lambda_j(t_{k+1}+\tau_m)\right)\\ 
                                                 &+0.5 M_1\left(\Lambda_j(t_{k+1}+\tau_m)\right) \\  
                                                &\leqslant 0.5\left(M^1_B-\epsilon_1\right)+0.5 (M^1_B+\delta(t_k)) \\  
                                                &=M^1_B-0.5 \epsilon_1 + \delta(t_k).
    \end{split}
\end{equation*}
Similar to (\ref{upperbound_j_k+1}), we can derive
\begin{align*}
    x^1_j(t_{k+2}+\tau_m+\tau_{m+1}) \leqslant M^1_\mathcal{B}-\epsilon_2,
\end{align*}
where $\epsilon_2 \triangleq (0.5 c)^2 \epsilon_0-\delta(t_k)-\delta(t_{k+1})$. Also, we can find that the upper bound can also be applied to agent $i \in \mathcal{B} \setminus \mathcal{B}_m^1\left(t_{k+1}+\tau_m, \sigma_M\right)$. The previous analysis indicates either or both of the following relations hold: 
\begin{itemize}
    \item $\mathcal{B}_M^1\left(t_{k+2}+\tau_m+\tau_{m+1}, \sigma_M\right) \subseteq \mathcal{B}_M^1\left(t_{k+1}+\tau_m, \sigma_M\right)$;
    \item $\mathcal{B}_m^1\left(t_{k+2}+\tau_m+\tau_{m+1}, \sigma_m\right) \subseteq \mathcal{B}_m^1\left(t_{k+1}+\tau_m, \sigma_m\right)$.
\end{itemize}
From \cite{yan2022resilient}, if the network is only under agent-based attacks, the agreement will be achieved exponentially with the ``auxiliary point" method. Since after at most $|\mathcal{B}|$ steps, either or both of the sets $\mathcal{B}_M^1(t_k, M_{\mathcal{B}}^1\left(t_k\right)-\epsilon_0)$ and $\mathcal{B}_m^1(t_k, m_{\mathcal{B}}^1\left(t_k\right)+\epsilon_0)$ will shrink to an empty set from the recursive process above. Our proposed policy treats the DoS attacks as delays in the communication links. We denote $t_{\eta}$ as a sufficient large instant which means that the duration $\xi_\eta$ between $t_{k}$ and $t_{\eta}$ includes efficient $|\mathcal{B}|$ steps and time delays caused by DoS attacks, namely, $\xi_\eta \geq |\mathcal{B}|+|\Pi_{D} (t_{k}, t_{\eta})|$, as illustrated in Fig. 2. Define $\epsilon_\eta \triangleq (0.5c)^{\eta}\epsilon_0-\sum_{t=t_k}^{t_k+t_\eta-1} \delta(t)$, where $\eta > 1$. 
\begin{figure}[h]
	\centering
	\includegraphics[scale=0.3]{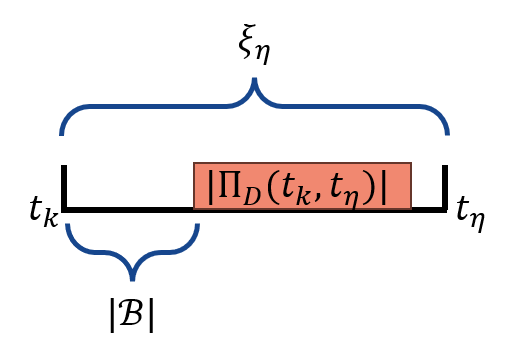}
	\caption{An explanation of time duration $\xi_\eta$, where the red area denoted the finite DoS attacks time between $t_k$ and $t_\eta$. $|\mathcal{B}|$ represents the effective updating steps.}
	\label{time duration for a single process}
\end{figure}
The vanishing property of two sets directly implies that one of the following relations must be true:
\begin{align} \label{empty_B_1}
	 \mathcal{B}_M^1\left(t_\eta,  M_{\mathcal{B}}^1\left(t_k\right)-\epsilon_0\right)=\emptyset, 
\end{align}
or,
\begin{align} \label{empty_B_2}
	\mathcal{B}_m^1\left(t_\eta,  m_{\mathcal{B}}^1\left(t_k\right)+\epsilon_0\right)=\emptyset,
\end{align}
or both. 
Direct results of (\ref{empty_B_1}) and (\ref{empty_B_2}) are as follows
\begin{align}
    M^1\left(t_{\eta}\right) \leq M^1\left(t_k\right)-\epsilon_{\eta},
\end{align}
or,
\begin{align}
    m^1\left(t_{\eta}\right) \geq m^1\left(t_k\right)+\epsilon_{\eta}.
\end{align}
Combining the two inequalities above, yields
\begin{equation}\label{Delta_vanish}
    \begin{split}
        \Delta^1\left(t_{\eta}\right) & \leq \Delta^1\left(t_k\right)-2\epsilon_\eta, \\
                                      & = (1-(0.5c)^{\eta})\Delta^1(t_k)+2 \sum_{t=t_k}^{t_k+t_{\eta}-1} \delta(t)
    \end{split}
\end{equation}
Now we do some analysis in order to apply Lemma 7 in \cite{nedic2010constrained}.
Let $\iota$ be an arbitrary positive integer. We consider the instant $t_{\iota\eta}$ and rearrange (\ref{Delta_vanish}) as
\begin{equation}\label{l_Delta}
    \begin{split}
        \Delta^1\left(t_{\iota\eta}\right) & \leq (1-(0.5c)^{\eta})^{\iota}\Delta(t_k) + \\
        & 2 \sum_{l=0}^{\textcolor[RGB]{170.00,0.0,0.00}{\iota-1}}(1-(0.5c)^\eta)^{\iota-1-l}\sum_{t=t_k+(\iota-1-l)\textcolor[RGB]{170.00,0.0,0.00}{t_{\eta}}}^{t_k+(\iota-l)\textcolor[RGB]{170.00,0.0,0.00}{t_{\eta}}-1} \delta(t).
    \end{split}
\end{equation}
Since (\ref{l_Delta}) can be considered as $\iota$ combinations of (\ref{Delta_vanish}), the time duration between $t_k$ and $t_{\iota\eta}$ can be explained in Fig. \ref{time duration for multi process}, which shows that Algorithm 1 can tolerate the DoS attacks satisfying after reaching the agreement. In this way, as illustrated in Assumption \ref{DoS-duration}, we only need to ensure that the duration of DoS attacks is finite within a certain range.
\begin{figure}[h]
	\centering
	\includegraphics[width=1.0\linewidth]{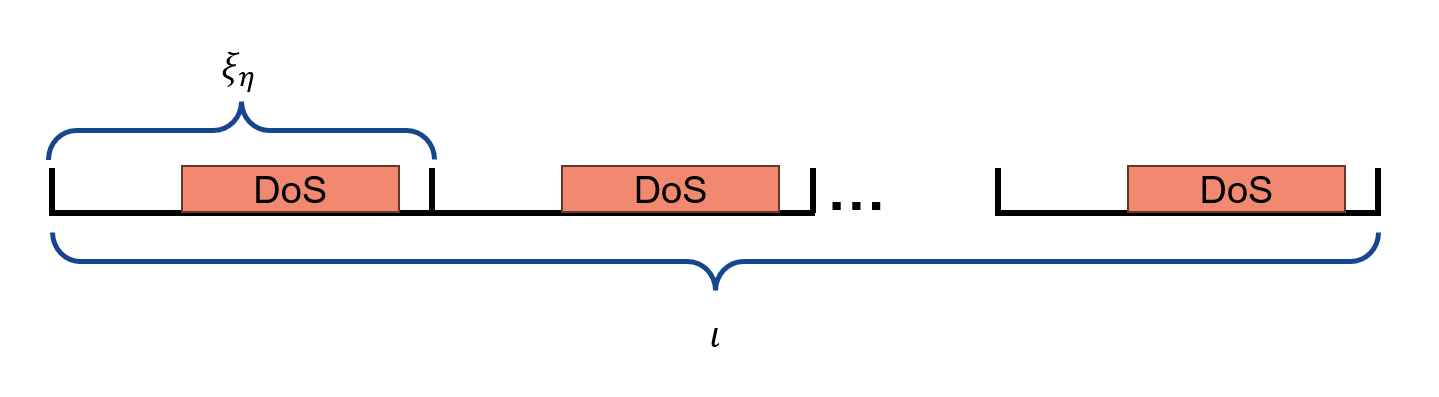}
	\caption{An explanation of time duration between $t_k$ and $t_{\iota\eta}$, where the red area denotes the independent finite DoS attacks time satisfying Assumption \ref{DoS-duration}. There is no requirement for general intervals to have consistent length.}
	\label{time duration for multi process}
\end{figure}

With Assumption \ref{residual assumption}, we have that
\begin{align} \label{lim = 0}
    \lim_{\iota \rightarrow \infty}\sum_{t=t_k+(\iota-1-l)\eta}^{t_k+(\iota-l)\eta-1} \delta(t) = 0.
\end{align}
According to Lemma 7 in \cite{nedic2010constrained}, if (\ref{lim = 0}) is guaranteed, the following asymptotic agreement along the first entry can be derived from (\ref{l_Delta}):
\begin{align*}
    \lim_{\iota \rightarrow \infty} \Delta^1(t_{\iota\eta}) = 0.
\end{align*}
Owing to the symmetric characteristic along each dimension, the proof is completed.

\end{proof}
Based on Lemmas \ref{vanish} and \ref{lemma agreement}, we directly conclude the problem of the resilient consensus against agent-based and DoS attacks is solved, which is summarized in the following theorem.
\begin{theorem}
	Consider an MAS cooperating over a \textbf{(d+1)F+1}-robust [\textbf{(dF+1, F+1)}-robust] digraph $\mathcal{G}=\{\mathcal{V}, \mathcal{E}\}$ under the F-local [F-total \textbf{resp.}] agent-based attacks. If Assumptions 1-3 are satisfied, Algorithm 1 guarantees that benign agents asymptotically achieve the resilient multi-dimensional consensus.
\end{theorem}

\begin{remark}
    Based on Definitions 1 and 2, it can be observed that a $((d + 1)F + 1)$-robust digraph is also $(dF + 1, F + 1)$-robust, but the reverse does not hold. In other words, a network capable of withstanding F-local attacks can also endure F-total attacks; however, the opposite is not necessarily true. This aligns with the understanding that F-total attacks represent a specific case of F-total attacks.
\end{remark}

\section{Resilient multi-dimensional distributed optimization}
\subsection{Algorithm design}
When each agent's local cost function exhibits multidimensionality, addressing resilient distributed optimization problems under agent-based attacks becomes significantly complex. Furthermore, the complexity escalates when combining agent-based and DoS attacks. From the problem formulation, the core distinction between the resilient consensus and optimization is that we need to ensure the final agreement states lie in the optimal set $\mathcal{X}^*$. We modify Algorithm 1 by incorporating a subgradient descent term to guarantee that the agreement point converges to the global minimizer.

\begin{algorithm}
\caption{RMDO Algorithm}
\label{Al_2}
\begin{algorithmic}[1]
\STATE Steps 1-3 are the same as those in Algorithm \ref{al1} to calculate the ``auxiliary point" $\tilde{x}_i(t_k)$.
\STATE Step 4: Control input of agent $i\in \mathcal{B}$ is designed as:
\begin{equation} \label{subgradient_update}
    u_i(t_k) = (1-\alpha_i)(\tilde{x}_i(t_k)-x_i(t_k))-\beta_{i,k} d_i(t_k),
\end{equation}
where the weight is chosen as $\alpha_i < c$, $1-\alpha_i < c$, and constant $0.5 \leq c <1$. $d_i(t_k)$ denotes the subgradient of the local cost function $f_i$ at $\alpha_ix_i(t_k)+(1-\alpha_i)\tilde{x}_i(t_k)$.
\STATE Step 5: Transmit the updated state $x_i(t_{k+1})$ to all out-neighbors $j\in \mathcal{N}_i^{out}$.

\end{algorithmic}
\end{algorithm}
\begin{remark}
    After utilizing the resilient consensus controller, the updating state is supposed to be $\alpha_ix_i(t_k)+(1-\alpha_i)\tilde{x}_i(t_k)$. To drive the final agreement toward the set of global optimizers, the subgradient is selected at the above point, taking into account that some local cost functions may not be differentiable. Compared with existing works \cite{gupta2021byzantine, yemini2022resilience}, our proposed method does not require a fully connected network or extra prior knowledge.
\end{remark}
The subgradient method widely adopts the following standard assumptions:
\begin{assumption} \label{bound of subgradient}
    The subgradient of each local cost function $f_i(x)$, $i\in \mathcal{V}$ is bounded. Specifically, given a positive scalar $L$, it holds that
    \begin{align*}
        \|d_i(t_k)\| \leq L,
    \end{align*}
    where $d_i(t_k)\in \partial f_i, i\in \mathcal{V}$.
\end{assumption}

\begin{assumption} \label{step size}
    The subgradient step-size sequence $\{\beta_{i,k}\}$ is non-negative, non-increasing and satisfies: $\lim_{k \rightarrow \infty} \beta_{i,k} = 0$, $\sum_{t_0}^{\infty}\beta_{i,k}=\infty$ and $\sum_{t_0}^{\infty} \beta_{i,k}^2 < \infty$.
\end{assumption}
\subsection{Performance analysis}
This paper considers a fully distributed optimization problem, that is, benign agents have no prior knowledge of the identities of adversarial agents. Under this setting, the objective redundancy can be used to ensure the following two optimization objectives are solved simultaneously:
\begin{enumerate} [1)]
    \item $\min \frac{1}{|\mathcal{B}|}f_{\mathcal{B}}(\boldsymbol{x}) = \min\frac{1}{|\mathcal{B}|} \sum_{i\in\mathcal{B}}f_i(\boldsymbol{x})$;
    \item $\min \frac{1}{N}f_{N}(\boldsymbol{x}) = \min\frac{1}{N} \sum_{i\in\mathcal{V}}f_i(\boldsymbol{x})$.
\end{enumerate}
The set of the optimal states $\mathcal{X}^*$ satisfies the following assumption.
\begin{assumption} \label{X* property}
    The set of optimal solutions $\mathcal{X}^*$ is bounded and nonempty.
\end{assumption}
Corollary 3 presented in \cite{zhu2023resilient} shows that if an MAS with $N$ agents is $F$-redundant for $F \geq 1$, satisfying Assumption \ref{X* property}, then for all $i\in \mathcal{V}$, the following two relations hold:
    \begin{align} \label{Corollary 3 for two targets}
        \begin{gathered}      
           \mathcal{X}^* \subset \underset{\boldsymbol{x}}{\arg \min } f_i(\boldsymbol{x}),\\
           \bigcap_{i \in \mathcal{V}} \underset{\boldsymbol{x}}{\arg \min } f_i(\boldsymbol{x})=\mathcal{X}^* .
        \end{gathered}
    \end{align}
The relations in (\ref{Corollary 3 for two targets}) demonstrate the equivalence of the two aforementioned
optimization problems.
We have the relation (\ref{safe kernel in Xi(t_k)}) which implies that the ``safe kernel" will never update outside the convex hull formed by the benign agents' states. Combining this relationship with the uniform bounded property (\ref{residual}) of the residual term, we can rewrite the updating rule (\ref{subgradient_update}) in Algorithm 2 as follows
\begin{align} \label{change_update}
    x_i(t_{k+1}) = \sum_{j\in \mathcal{B}} \Bar{a}_{ij}(t_k)x_j(t_k) - \beta_{i,k} d_i(t_k),
\end{align}
where $\Bar{a}_{ij}(t_k) \geq 0$ and $\sum_{j\in \mathcal{B}}\Bar{a}_{ij}(t_k) = 1$. For any instants $t_k \geq 0$, the row-stochastic matrix is defined as $\Bar{\mathcal{A}}(t_k)=(\Bar{a}_{ij}(t_k)) \subset \mathbb{R}^{|\mathcal{B}| \times |\mathcal{B}|}$. 

\begin{lemma}
    Consider an MAS cooperating over the \textbf{(d + 1)F + 1}-robust [\textbf{(dF + 1, F + 1)}-robust] digraph subject to the agent-based attacks which conform to the F-local [F-total \textbf{resp.}] attack model. If Assumption \ref{No of in neighbors} is satisfied, the row-stochastic matrix $\Bar{\mathcal{A}}(t_k)$ is Sarymsakov as Definition \ref{Sary}.
\end{lemma}

\begin{proof}
    To show $\Bar{\mathcal{A}}(t_k)$ is Sarymsakov, we consider two disjoint and nonempty subsets, $\mathcal{B}_1, \ \mathcal{B}_2 \subsetneq \mathcal{B}$  and an agent $i \in \mathcal{B}$. With Assumption \ref{No of in neighbors}, one of the following statements must hold:
\begin{enumerate}
    \item For $j \in \mathcal{F}_{\Bar{\mathcal{A}}(t_k)}(\mathcal{B}_2)$, $\mathcal{F}_{\Bar{\mathcal{A}}(t_k)}(\mathcal{B}_1) \cap \mathcal{F}_{\Bar{\mathcal{A}}(t_k)}(\mathcal{B}_2) = \emptyset$;
    \item For $j \notin \mathcal{F}_{\Bar{\mathcal{A}}(t_k)}(\mathcal{B}_2)$, $\mathcal{F}_{\Bar{\mathcal{A}}(t_k)}(\mathcal{B}_1) \cap \mathcal{F}_{\Bar{\mathcal{A}}(t_k)}(\mathcal{B}_2) = \emptyset$ holds. With the property of $((d+1)F+1)$-robustness [$(dF+1,\ F+1)$-robustness], we can conclude that $|\mathcal{F}_{\Bar{\mathcal{A}}(t_k)}(\mathcal{B}_1) \cup \mathcal{F}_{\Bar{\mathcal{A}}(t_k)}(\mathcal{B}_2)| \geq |\mathcal{B}_1|+|\mathcal{B}_1|+1 >|\mathcal{B}_1 \cup \mathcal{B}_2|$. 
\end{enumerate}
From Definition \ref{Sary}, it has been proved that $\Bar{\mathcal{A}}(t_k)$ is Sarymsakov.
\end{proof}

Furthermore, a useful definition is introduced as follows.
\begin{definition}
    \cite{kolmogoroff1936theorie} A stochastic vectors ${\boldsymbol{\pi}(t_k)}$, where $\sum_i \pi_i(t_k) = 1$ and $\pi_i(t_k) \geq 0$ is said to be an absolute probability sequence for a sequence of row-stochastic matrices $\mathcal{A}(t_k)$ if
    \begin{align*}
        \boldsymbol{\pi}^T(t_{k+1}) \mathcal{A}(t_k) = \boldsymbol{\pi}^T(t_k), \ \forall t_k \geq 0.
    \end{align*}
\end{definition}
It has been proved in \cite{kolmogoroff1936theorie} that the absolute probability sequence always exists for any row-stochastic matrices. Therefore, we can always find a sequence of stochastic vectors $\{\boldsymbol{\pi}(t_k)\}$ such that
\begin{align} \label{one-time jump}
    \boldsymbol{\pi}^T(t_{k+1})\Bar{\mathcal{A}}(t_k) = \boldsymbol{\pi}^T(t_{k}), \ \forall t_k \geq 0.
\end{align}
Multiplying $\Bar{\mathcal{A}}(t_{k-1})$, $\Bar{\mathcal{A}}(t_{k-2})$, ..., $\Bar{\mathcal{A}}(t_{0})$ recursively on the left hand of (\ref{one-time jump}) yields
\begin{align}\label{jump from 0 to t_k}
    \boldsymbol{\pi}^T(t_{k+1})\Bar{\mathcal{A}}(t_k, t_0) = \boldsymbol{\pi}^T(t_0),
\end{align}
where $\Bar{\mathcal{A}}(t_{k2}, t_{k1}) \triangleq \Bar{\mathcal{A}}(t_{k1})\Bar{\mathcal{A}}(t_{k1-1}) \dots \Bar{\mathcal{A}}(t_{k2})$, $t_{k1} > t_{k2}$.

To investigate benign states as time goes to infinity, we define $\pi_\infty \triangleq \lim_{t_k \rightarrow \infty} \pi(t_k)$. The following lemma is established to determine the existence set of optimal solutions as $t_k$ goes to $\infty$.
\begin{lemma} \label{pi(infty) exists}
    Consider an MAS cooperating over the \textbf{(d + 1)F + 1}-robust [\textbf{(dF + 1, F + 1)}-robust] digraph subject to DoS attacks. Agent-based attacks conform to the F-local [F-total \textbf{resp.}] attack model. If Assumptions 1-6 are satisfied:
	then $\pi_\infty = [\pi_{\infty 1},\ \dots,\ \pi_{\infty |\mathcal{B}|}]$ exists with the application of Algorithm 2. 
\end{lemma}

\begin{proof}
    Since the set of Sarymsakov matrices is closed under matrix multiplication and any $\Bar{\mathcal{A}}(t_k)$ is Sarymsakov, we can say that $\Bar{\mathcal{A}}(t_k, t_0)$ belongs to Sarymsakov class. Referred to Theorem 1 of \cite{xia2014sarymsakov}, there always exists a column stochastic vector $\pi = [\pi_1,\dots, \pi_N]$ with $\pi_i \geq 0$ and $\sum_{i=1}^{N} \pi_i=1$, such that
    \begin{align}
        \lim_{t_k \rightarrow \infty} \Bar{\mathcal{A}}(t_k, t_0) = \mathbf{1} \boldsymbol{\pi}^T.
    \end{align}
    Combined with (\ref{jump from 0 to t_k}), Lemma \ref{pi(infty) exists} is proved.
\end{proof}

Let $\mathcal{X}^*_{\pi} \triangleq \underset{x}{\arg \min } \sum_{i \in \mathcal{B}} \pi_{\infty}f_i(x)$. We have the following lemma for the relation between $\mathcal{X}^*$ and $\mathcal{X}^*_{\pi}$.
\begin{lemma} \label{X and X pi relation}
    If an MAS consisting $N$ agents is $F$ redundant and Assumption \ref{X* property} is satisfied, the following relation holds
    \begin{align}\label{X and X pi}
        \mathcal{X}^* = \bigcup_{\pi_\infty} \mathcal{X}^*_{\pi}.
    \end{align}
\end{lemma}
\begin{proof}
    Lemma \ref{pi(infty) exists} guarantees the existence of $\mathcal{X}^*_{\pi}$. Moreover, from Lemma \ref{r-redundant}, $F$-redundant digraph guarantees that any subsets $\mathcal{V}' \in \mathcal{V}$ with cardinality $|\mathcal{V}'| \geq |N-F|$ satisfying $\underset{x}{\arg \min } \sum_{i \in \mathcal{V}'(t_k)} f_i(x) = \mathcal{X}^*$. Combining the fact that $\sum_{i=1}^{N} \pi_{\infty i}=1$, we can derive $\mathcal{X}^* = \mathcal{X}^*_{\pi}$ when $\pi_\infty$ is unique. However, considering the non-uniqueness of $\pi_\infty$, we can further infer the relation given in (\ref{X and X pi}).
\end{proof}

Lemma \ref{X and X pi relation} indicates that the optimization target (\ref{optimization target}) will be achieved if all benign agents' states stay within the set $\mathcal{X}^*_{\pi}$ as time goes to infinity. Based on this idea, Theorem 2 is established as follows:
\begin{theorem}
   Consider an MAS cooperating over a \textbf{(d+1)F+1}-robust [\textbf{(dF+1, F+1)}-robust] digraph $\mathcal{G}=\{\mathcal{V}, \mathcal{E}\}$ under the F-local [F-total \textbf{resp.}] agent-based attacks. If Assumptions 1-6 are satisfied and local cost functions are $F$-redundant, Algorithm 2 guarantees that benign agents asymptotically achieve the RMDO.
\end{theorem}

\begin{proof}
     We now commence with the proof of agreement. Similar to (\ref{Delta_vanish}), and with the subgradient term, we have
\begin{equation}\label{vanish-delta-sub}
    \begin{split}
      \Delta^1\left(t_{\iota\eta}\right) & \leq (1-(0.5c)^{\eta})^{\iota}\Delta(t_k) + 2L\iota\eta\beta_{i,\iota\eta}\\
        & 2 \sum_{l=0}^{\iota-1}(1-(0.5c)^\eta)^{\iota-1-l}\sum_{t=t_k+(\iota-1-l)t_{\eta}}^{t_k+(\iota-l)t_{\eta}}\delta(t).
    \end{split}
\end{equation}
    Given the fact that $\beta_{i,\iota\eta} \rightarrow 0$ as $\iota$ goes to $\infty$, (\ref{vanish-delta-sub}) is the same as (\ref{Delta_vanish}) as $\iota$ goes to $\infty$. We can say the agreement can be achieved asymptotically.
    
    Before the proof of the optimization target (\ref{optimization target}), some notations are given: $ v_i(t_k) = \sum_{j\in \mathcal{B}} \Bar{a}_{ij}x_j(t_k)$ and $\Bar{x}(t_k) = \frac{1}{|\mathcal{B}|} \sum_{i \in \mathcal{B}} x_i(t_k)$.
    The distance between a point $x$ and a set $\mathcal{S}$ is defined as:
    \begin{align*}
        d_{\mathcal{S}}(x) \triangleq  \inf_{y \in \mathcal{S}}\|x-y\|^2.
    \end{align*}
    For any benign agent, we denote an arbitrary point in $\mathcal{X}^*_{\pi}$ as $x^* = \sum_{j\in \mathcal{B}} \Bar{a}_{ij}(t_k)\mathcal{P}_{\mathcal{X}^*_{\pi}}(x_j(t_k))$, where $\mathcal{P}_{\mathcal{X}^*_{\pi}}(\cdot)$ is the projection operator.
    We still consider the worst case for the benign agent $i$, that is, $i$'s in-neighboring states are interrupted by $m$-th DoS attacks. We can rewrite (\ref{change_update}) as follows
    \begin{equation*}
        \begin{split}
            x_i(t_{k+1}+\tau_m) &= \sum_{j\in \mathcal{B}} \Bar{a}_{ij}x_j(t_k) - \beta_{i,k} d_i(t_k) \\
                                &= v_i(t_k) - \beta_{i,k} d_i(t_k).
        \end{split}
    \end{equation*}
    Then we investigate the distance from $x_i(t_{k+1}+\tau_m)$ to the optimal set $\mathcal{X}_\pi^*$:
    \begin{equation*}
        \begin{split}
            &d_{\mathcal{X}_{\pi}^*}(x_i(t_{k+1}+\tau_m)) \\
            &= \|x_i(t_{k+1}+\tau_m)-x^*\|^2\\
            &=\|v_i(t_k) - \beta_{i,k} d_i(t_k)-x^*\|^2\\
            &=\|v_i(t_k) -x^*\|^2 + \beta_{i,k}^2\|d_i(t_k)\|^2 - 2\beta_{i,k}  \langle d_i(t_k), v_i(t_k)-x^* \rangle
        \end{split}
    \end{equation*}
    From the property of the subgradient, $d_i(t_k)(v_i(t_k)-x^*) \geq f_i(v_i(t_k))-f_i(x^*)$ holds. Combined with the convexity of the squared norm and Assumption \ref{bound of subgradient}, we derive the following inequalities:
    \begin{equation*}
        \begin{split}
         &\quad \ \|v_i(t_k) -x^*\|^2 + \beta_{i,k}^2\| d_i(t_k)\|^2 + 2\beta_{i,k} \langle d_i(t_k), x^*-v_i(t_k) \rangle \\
         & \leq   \sum_{j\in \mathcal{B}} \Bar{a}_{ij}(t_k)\|x_j(t_k) -x^*\|^2 +\beta_{i,k}^2L^2 \\ 
         & \quad  -2\beta_{i,k}[f_i(v_i(t_k))-f_i(x^*)] \\
         &=\sum_{j\in \mathcal{B}} \Bar{a}_{ij}(t_k)d_{\mathcal{X}_{\pi}^*}(x_j(t_k)) +\beta_{i,k}^2 L^2  \\  &\quad  +2\beta_{i,k}[f_i(\Bar{x}(t_k))-f_i(v_i(t_k))]-2\beta_{i,k}[f_i(\Bar{x}(t_k))-f_i(x^*)] \\
         &\leq \sum_{j\in \mathcal{B}} \Bar{a}_{ij}(t_k) d_{\mathcal{X}_{\pi}^*}(x_j(t_k)) +\beta_{i,k}^2L^2 + 2\beta_{i,k} L [\Bar{x}(t_k)-v_i(t_k)] \\
         &\quad  + 2\beta_{i,k}[f_i(\Bar{x}(t_k))-f_i(x^*)].
        \end{split}
    \end{equation*}
    Multiply $\pi_i(t_{k+1}+\tau_m)$ and add up the above relations for all $i \in \mathcal{B}$. Owing to  (\ref{one-time jump}) and the definition of the stochastic vector, $\pi_j(t_k)=\sum_{i \in \mathcal{B}} \Bar{a}_{ij}(t_k)\pi_i(t_{k+1}+\tau_m)$ and $\sum_{i \in \mathcal{B}}\pi_i(t_{k+1}+\tau_m)\beta_{i,k}^2L^2 = \beta_{i,k}^2L^2$ hold. Then we have the following inequality:
    \begin{equation*}
        \begin{split}
            & \quad \ \sum_{i \in \mathcal{B}}\pi_i(t_{k+1}+\tau_m) d_{\mathcal{X}_{\pi}^*}(x_i(t_{k+1}+\tau_m)) \\ 
            &\leq \sum_{j\in \mathcal{B}} \pi_j(t_k) d_{\mathcal{X}_{\pi}^*}(x_j(t_k)) +\beta_{i,k}^2L^2 \\ 
            & \quad \ + 2\beta_{i,k} L \sum_{i \in \mathcal{B}}\pi_i(t_{k+1}+\tau_m)[\Bar{x}(t_k)-v_i(t_k)]\\
            &\quad \  - 2\beta_{i,k}\{\sum_{i \in \mathcal{B}}\pi_i(t_{k+1}+\tau_m)[f_i(\Bar{x}(t_k))-f_i(x^*)]\}.
        \end{split}
    \end{equation*}
    Define $D_{t_k}^2 \triangleq \sum_{i \in \mathcal{B}}\pi_i(t_k) d_{\mathcal{X}_{\pi}^*}(x_i(t_k))$, we can rewrite the above inequality as
    \begin{equation}\label{D_k inequality}
        \begin{split}
            &2\beta_{i,k}\{\sum_{i \in \mathcal{B}}\pi_i(t_{k+1}+\tau_m)[f_i(\Bar{x}(t_k))-f_i(x^*)]\} \\ &\leq  D_{t_k}^2 -D_{t_{k+1}+\tau_m}^2+ \beta_{i,k}^2L^2 \\
                                      & + 2\beta_{i,k} L \sum_{i \in \mathcal{B}}\pi_i(t_{k+1}+\tau_m)[\Bar{x}(t_k)-v_i(t_k)].
        \end{split}
    \end{equation}
    Let $\theta$ be a sufficiently large integer such that $t_\theta$ includes the efficient updating instants and all DoS attacks' duration. Recursively adding up both sides of (\ref{D_k inequality}) from $t_0$ to $t_\theta$, we have the following inequalities:
    \begin{equation}\label{D_theta inequality}
        \begin{split}
            &2\sum_{t_0}^{t_\theta}\beta_{i,k}\{\sum_{i \in \mathcal{B}}\pi_i(t_{k+1})[f_i(\Bar{x}(t_k))-f_i(x^*)]\} \\ 
            &\leq  D_{t_0}^2 -D_{t_{\theta}}^2+ L^2 \sum_{t_0}^{t_{\theta-1}}\beta_{i,k}^2 \\
            & + 2 L \sum_{t_0}^{t_{\theta-1}} \beta_{i,k} \sum_{i \in \mathcal{B}}\pi_i(t_{k+1})[\Bar{x}(t_k)-v_i(t_k)].
        \end{split}
    \end{equation}
    As $t_\theta$ goes to $\infty$, we can rewrite (\ref{D_theta inequality}) as follows
    \begin{equation}\label{D_infty inequality}
        \begin{split}
            &2\sum_{t_0}^{\infty}\beta_{i,k}\{\sum_{i \in \mathcal{B}}\pi_i(t_{k+1})[f_i(\Bar{x}(t_k))-f_i(x^*)]\} \leq  D_{t_0}^2 +\\ 
            &\quad  L^2 \sum_{t_0}^{\infty}\beta_{i,k}^2 
             + 2 L \sum_{t_0}^{\infty} \beta_{i,k} \sum_{i \in \mathcal{B}}\pi_i(t_{k+1})[\Bar{x}(t_k)-v_i(t_k)].
        \end{split}
    \end{equation}
    The proved agreement condition and $\lim_{k \rightarrow \infty} \beta_{i,k} = 0$ in Assumption \ref{step size} imply $2 L \sum_{t_0}^{\infty} \beta_{i,k} \sum_{i \in \mathcal{B}}\pi_i(t_{k+1})[\Bar{x}(t_k)-v_i(t_k)] < \infty$. With $\sum_{t_0}^{\infty} \beta_{i,k}^2 < \infty$ in Assumption \ref{step size}, we can claim that $\sum_{t_0}^{\infty}\beta_{i,k}\{\sum_{i \in \mathcal{B}}\pi_i(t_{k+1})[f_i(\Bar{x}(t_k))-f_i(x^*)]\} < \infty$. Since $\sum_{t_0}^{\infty}\beta_{i,k}=\infty$, the above inequality indicates
    \begin{align}
        \lim_{t_k \rightarrow \infty} \inf \sum_{i \in \mathcal{B}}\pi_i(t_{k+1})[f_i(\Bar{x}(t_k))-f_i(x^*)] =0,
    \end{align}
    which entails
    \begin{align}
        \lim_{t_k \rightarrow \infty} \inf d_{\mathcal{X}_\pi^*}(\Bar{x}(t_k))=0.
    \end{align}
    Combined with the agreement condition and Lemma \ref{X and X pi relation}, $\lim_{t_k \rightarrow \infty} \sum_{i \in \mathcal{B}}\|\Bar{x}(t_k)-v_i(t_k)\|^2=0$. The optimization target (\ref{optimization target}) is guaranteed. 
\end{proof}

\begin{remark}
    The trade-off between the convergence rate and the robustness against external disturbances is related to the upper bound \(c\) of the weight \(\alpha_i\). Specifically, selecting \(c\) closer to 1 yields a faster convergence rate, whereas choosing \(c\) closer to 0.5 improves the robustness against external disturbances. This relationship is explicitly characterized by inequalities (\ref{Delta_vanish}) and (\ref{vanish-delta-sub}). These inequalities demonstrate that a larger \(c\) leads to a faster convergence rate but also amplifies the cumulative effect of disturbances over the finite time interval \([t_k, t_{\iota\eta}]\). As \(\iota \to \infty\), according to Assumption \ref{residual assumption}, the asymptotic consensus performance is ensured regardless of the specific choice of \(c\), provided that the weight selection satisfies the design condition.
\end{remark}

\section{Numerical examples}
Now, we utilize the following numerical examples to illustrate and verify the established results. The digraph $\mathcal{G}=\{\mathcal{V}, \mathcal{E}\}$ is with $[(d+1)F+1]-$robustness, where $d=2,F=2$. $x_i(1)$ and $x_i(2)$ represent the first and second entries of the state $x_i$. The agent-based attacks align with the $F$-local model. As shown in Fig. \ref{digraph}, agents $2$ and $3$ are adversarial, which will update their state as $x_2(1)=3 \sin t_k, x_2(2)=t_k+2$ and $x_{12}(1)=6\cos{t_k}, x_{12}(2)=3t_k$. The residual term is set to be $0.1^{t_k}$. And the weight $\alpha_i = 0.5$.  The communication
topology is as depicted in Fig. \ref{digraph}. It can be noted that all benign agents except agents $1$ and $7$ have $7$ in-neighbors, whereas agents $1$ and $7$ have $8$ in-neighbors to ensure the generality.
\begin{figure}[h]
	\centering
	\includegraphics[scale=0.7]{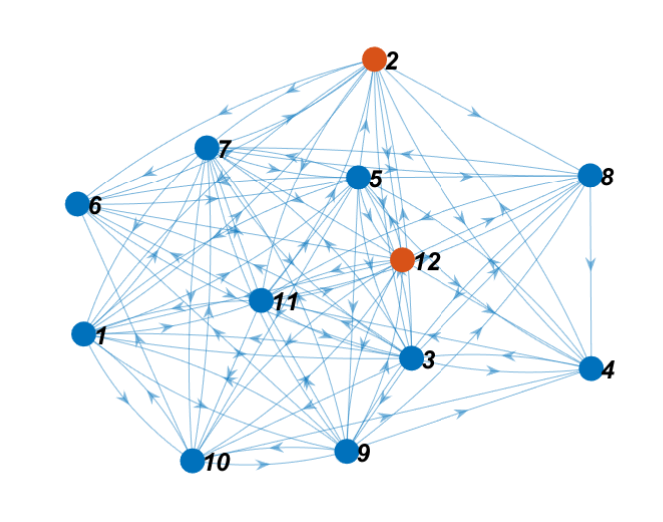}
	\caption{Communication network}
	\label{digraph}
\end{figure} 

We suppose the network is under the following 4 types of DoS attacks as shown in Fig. \ref{DDoS attacks types}. Different DoS attack types are launched on different edges to simulate the random behaviors of DoS attackers.

\begin{figure}[h]
	\centering
	\includegraphics[scale=0.55]{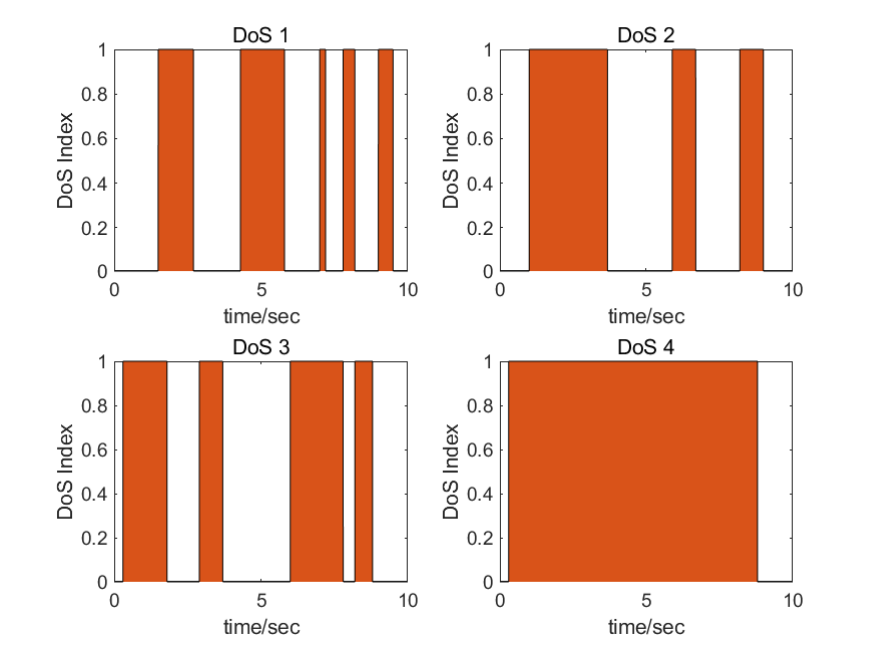}
	\caption{Types of DoS Attacks}
	\label{DDoS attacks types}
\end{figure}

Details for benign agents' in-neighboring agents' index, initial states, local cost functions, and edge-targeted DoS attacks are presented in Table \ref{tab:your_label}, where $(i,j)$ denotes the edge $e_{ij}$.
\begin{table*}[h!t]
\scriptsize
\centering
\begin{tabular}{|c|c|c|c|c|}
\hline
\textbf{Index} & \textbf{In-neighboring Nodes} & \makecell{\textbf{Initial} \\ \textbf{States}} & \textbf{DoS Target Edges} & \textbf{Cost functions} \\ \hline
1 & 2, 3, 5, 7, 9, 10, 11, 12 & [0;0] & DoS\_1: (1, 3), (1, 7), (1, 9) & \( f_1 = x_1(1)^2 + x_1(2)^2 \) \\ \hline
3 & 1, 2, 4, 5, 6, 7, 11 & [0;4] & DoS\_1: (3, 4), (3, 7), (3, 11) \& DoS\_3: (3, 2), (3, 6) & \( f_3 = |x_3(1)+x_3(2)| \) \\ \hline
4 & 2, 3, 5, 8, 9, 10, 12 & [2;8] & DoS\_2: (4, 2), (4, 3) & \( f_4 = \sqrt{x_4(1)^2+x_4(2)^2} \) \\ \hline
5 & 1, 3, 4, 6, 7, 11, 12 & [4;12] & DoS\_1: (5, 1), (5, 11)\& \ DoS\_3: (5, 3), (5, 7) & \( f_5 = (x_5(1)-1)^2+(x_5(2)-0.5)^2 \) \\ \hline
6 & 2, 3, 5, 7, 9, 10, 11 & [6;18] & DoS\_1: (6, 2), (6, 10) & \( f_6 = \max(x_6(1), x_6(2)) \) \\ \hline
7 & 1, 2, 3, 5, 8, 9, 10, 11 & [8;10] & DoS\_2: (7, 5), (7, 9), (7, 10)  \& DoS\_3: (7, 1), (7, 2), (7, 3) & \( f_7 = (x_7(1)+0.3)^2+(x_7(1)-0.2)^2 \) \\ \hline
8 & 2, 3, 5, 7, 9, 11, 12 & [10;7] & DoS\_2: (8, 3), (8, 5) & \( f_8 = |x_8(1)| + |x_8(2)| \) \\ \hline
9 & 1, 3, 5, 7, 8, 10, 12 & [10;6] & DoS\_2: (9, 7), (9, 8) & \( f_9 = |x_9(1)| + x_9(2)^2 \) \\ \hline
10 & 1, 3, 5, 7, 9, 11, 12 & [10;2] & DoS\_1: (10, 1), (10, 12) & \( f_{10} = x_{10}(1)^2+x_{10}(2)^2+|x_{10}(1)-x_{10}(2)| \) \\ \hline
11 & 1, 3, 4, 5, 8, 9, 12 & [10;0] & DoS\_3: (11, 1), (11, 3), (11, 12) & \( f_{11} = \max(|x_{11}(1)|, |x_{11}(2)|) \) \\ \hline
\end{tabular}
\caption{Configuration details of benign agents}
\label{tab:your_label}
\end{table*}
The sampling interval for all simulations is set to be $T=t_{k+1}-t_k=0.5\text{sec}$.

\subsection{Resilient multi-dimensional consensus}
The resilient consensus simulation result by applying Algorithm 1 is given in Fig. \ref{Algo 1 result}. The initial states of benign agents are labeled as $'\times'$ and those of faulty agents as $'\bigcirc'$. This result demonstrates the robustness of Algorithm 1 against both agent-based and DoS attacks. Furthermore, the effectiveness of our proposed policy can be verified by giving the result in Fig. \ref{Algo 1 yan result}.
\begin{figure}[h]
	\centering
	\includegraphics[scale=0.55]{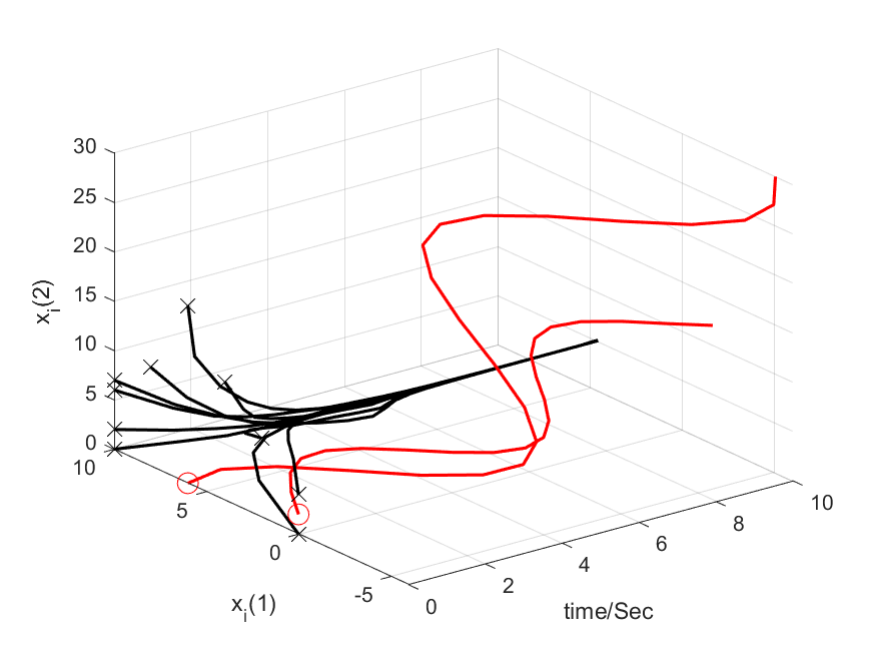}
	\caption{Trajectories of states by applying Algorithm 1 ($x_i(1)$ and $x_i(2)$ against time $t_k$).  The red curves denote the trajectories of malicious agents' states.}
	\label{Algo 1 result}
\end{figure}

\begin{figure}[h]
	\centering
	\includegraphics[scale=0.55]{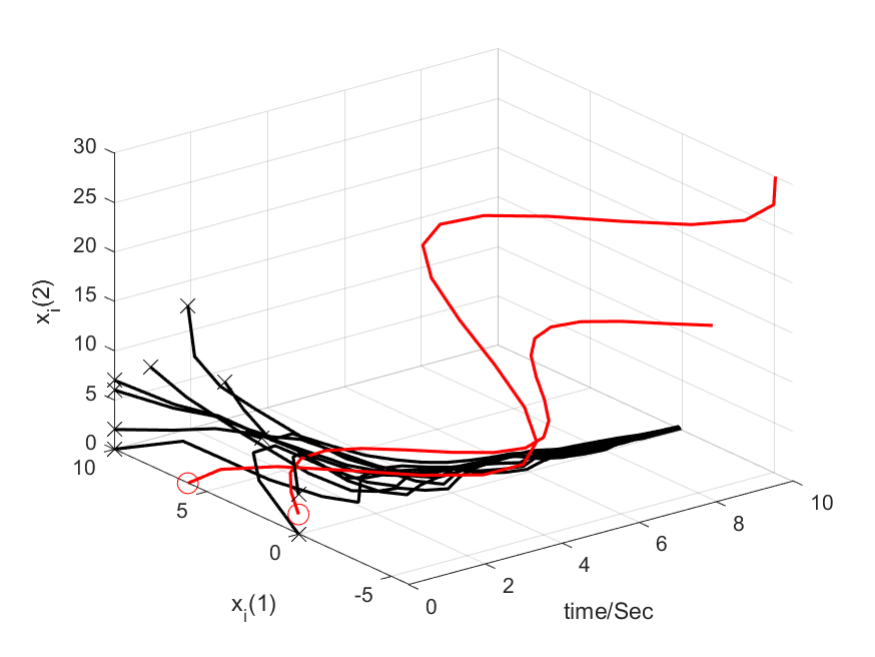}
	\caption{Trajectories of states by applying the algorithm in \cite{yan2022resilient} ($x_i(1)$ and $x_i(2)$ against time $t_k$).  The red curves denote the trajectories of malicious agents' states.}
	\label{Algo 1 yan result}
\end{figure}

From Fig. \ref{Algo 1 yan result}, the resilient consensus is not guaranteed, and all benign agents' states are driven to $[0;0]$. This phenomenon is mainly because if the communication links are blocked, no information is received from in-neighboring agents and thus their states become $[0;0]$ through these links. The large gaps between benign agents' states and $[0;0]$ will make some benign in-neighbors change into stubborn agents and remain at $[0;0]$ during DoS attacks. This may exceed the ``auxiliary point" method's maximum tolerance of faulty agents. To this end, our proposed policy is necessary.

\subsection{Resilient multi-dimensional distributed optimization}
The RMDO simulation settings are presented in Table \ref{tab:your_label}. The initial states, communication topology, agent-based and DoS attacks are the same as those of the resilient consensus simulation. The local cost functions of benign agents are convex, but some of them are not differentiable such as $f_3$ and $f_4$. Since the malicious agents do not update their states according to the prescribed algorithm, only their states will influence the achievement of resilient distributed optimization rather than their local cost functions. Therefore, omitting the local cost functions for agents 2 and 12 still satisfies the requirements of $k$-redundant. For simplicity and clarity, the step size is set as $\beta_{i,k} = 1/(5t_k + 1)$ for all healthy agents.  The trajectories of agents' states by applying Algorithm 2 are presented in Fig. \ref{Algo 2 result} and the value and changing rate of the global cost function are given in Fig. \ref{f global result}.

\begin{figure}[h]
	\centering
	\includegraphics[scale=0.55]{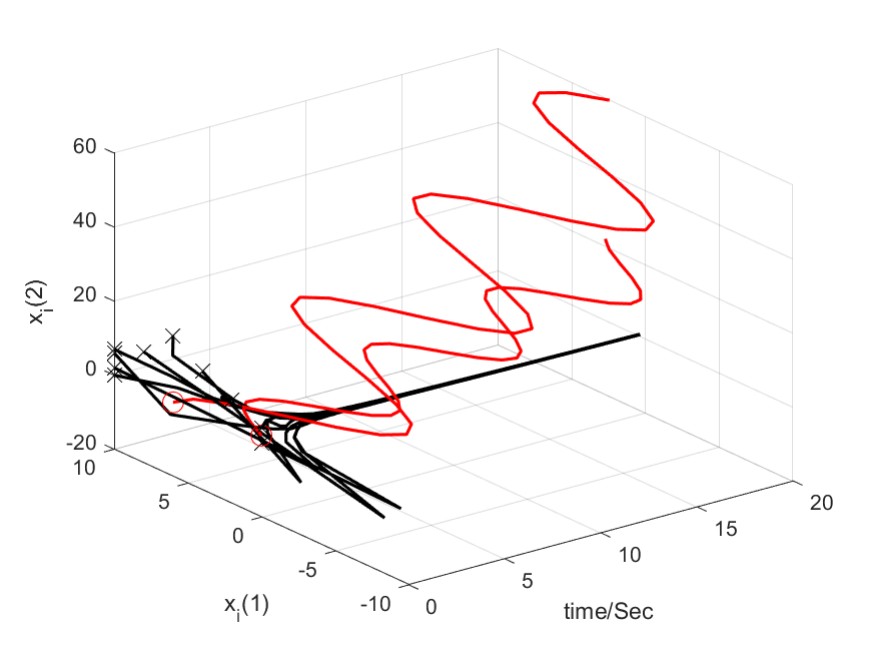}
	\caption{Trajectories of states by appling Algorithm 2  ($x_i(1)$ and $x_i(2)$ against time $t_k$).  The red curves denote the trajectories of malicious agents' states.}
	\label{Algo 2 result}
\end{figure}

\begin{figure}[h]
	\centering
	\includegraphics[scale=0.55]{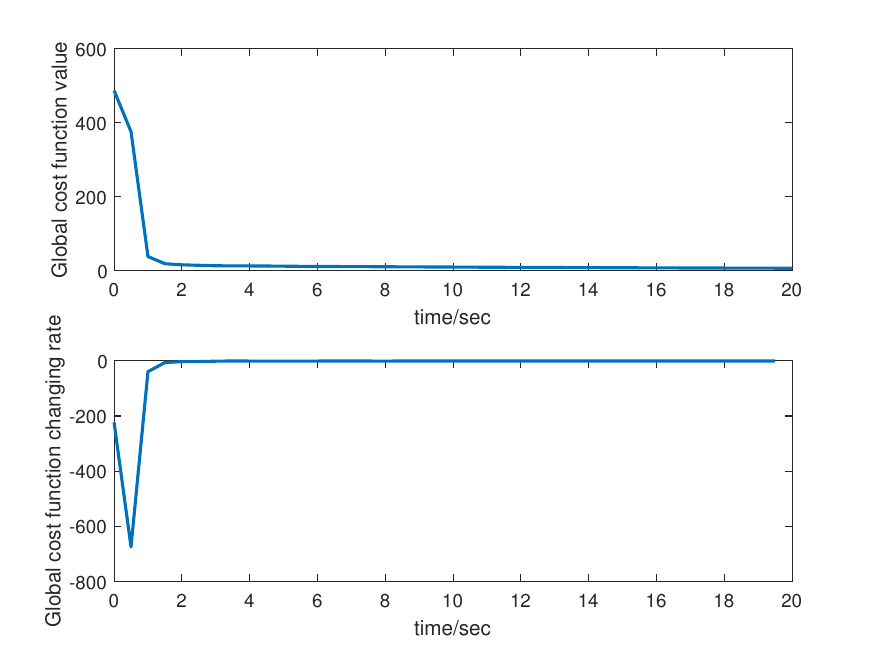}
	\caption{Global cost function value and its changing rate by applying Algorithm 2 against time $t_k$.}
	\label{f global result}
\end{figure}

From Fig. \ref{Algo 2 result} and Fig. \ref{f global result}, we can tell both the consensus condition (\ref{Consensus}) and optimization target (\ref{optimization target}) are achieved by applying Algorithm 2.  Owing to the complexity of the global cost function, directly computing its minimum is significantly difficult. Fig. \ref{f global result} indicates the global cost function value will settle at a minimum value asymptotically. Since the global cost function is convex, we can conclude that the exact convergence to the global minimizer has been achieved.

\section{Conclusion}
We propose resilient multi-dimensional consensus and distributed optimization algorithms for MASs against the influence of agent-based and edge-targeted DoS attacks. As for the resilient consensus algorithm, our proposed algorithm utilizes the immediate in-neighboring states before DoS attacks to enhance the robustness of the previous ``auxiliary point" method against both agent-based and edge-targeted attacks. Considering the challenges of RMDO problems, we combine the subgradient method with our proposed resilient consensus algorithm to achieve an exact convergence to the global optimizer under fully distributed settings. All theoretical results are established through mathematical analysis and verified by MATLAB simulation. 

\bibliography{manuscript} 
\bibliographystyle{ieeetr}

\end{document}

%% file: before_document.tex
\usepackage{graphicx}
\usepackage{subfig}
\usepackage{amsmath}
\usepackage{amsthm} 
\usepackage{amssymb}
\usepackage[font=footnotesize]{caption} 
\usepackage{subfig} 
\usepackage[noadjust]{cite} 
\usepackage{color}
\usepackage{algorithm} 
\usepackage{algorithmic}
\usepackage{enumerate}
\usepackage[hidelinks,colorlinks=false]{hyperref}
\usepackage[left=0.75in, right=0.75in, top=0.75in, bottom=0.75in]{geometry}
\usepackage{authblk} 
\usepackage{cite}

\usepackage{bm}
\usepackage{tikz}
\usetikzlibrary{calc} 
\usetikzlibrary{shapes} 
\usetikzlibrary{chains}
\usetikzlibrary{fit}
\usetikzlibrary{arrows}
\usetikzlibrary{decorations.text} 
\usetikzlibrary{decorations.markings}
\usetikzlibrary{decorations.pathmorphing} 
\usetikzlibrary{shadows}
\usetikzlibrary{patterns}
\usetikzlibrary{matrix}
\usepackage{pgfplots}
\usepackage[europeanresistors]{circuitikz}
\usepackage[outline]{contour} 
\contourlength{1.5pt}

\newtheorem{theorem}{Theorem}
\newtheorem{lemma}{Lemma}
\newtheorem{assumption}{Assumption}
\newtheorem{remark}{Remark}

\newtheorem{corollary}{Corollary}

\newtheorem{definition}{Definition}

\graphicspath{{figures/}}